\documentclass[11pt]{asaproc}
\usepackage{natbib}
\usepackage{amssymb,amsmath, amsthm, setspace, color, graphics, fullpage}
\usepackage{mathpazo, mathptmx, flexisym}
\usepackage{fancyhdr}
\usepackage{url}
\usepackage{breqn}
\usepackage{functan}
\usepackage[margin=1in]{geometry} 
\usepackage{times}
\usepackage[pdfborder={0 0 0},final=true,colorlinks=true,linkcolor=magenta,citecolor=blue]{hyperref}

\numberwithin{equation}{section}
\newcommand{\tab}{\hspace*{2em}}              
\newcommand{\indentitem}{\setlength\itemindent{25pt}} 
\renewcommand{\abstractname}{Abstract}
\newcommand{\bm}[1]{\mathbf{#1}}                     
\newcommand{\diadt}[2]{t^{(#1)}_{#2}}                
\newcommand{\vechat}[1]{\hat{\pmb{#1}}}                 

\Macro{xj}{\bm{x}_{\diadt{n}{j}}}                              
\Macro{xjplus1}{\bm{x}_{\diadt{n}{{j+1}}}}                              
\Macro{xTj}{\bm{x}^T_{\diadt{n}{j}}}                           
\Macro{xTjplus1}{\bm{x}^T_{\diadt{n}{{j+1}}}}                           
\Macro{zj}{\bm{z}_{\diadt{n}{j}}}                              
\Macro{zjplus1}{\bm{z}_{\diadt{n}{{j+1}}}}                              
\Macro{zTj}{\bm{z}^T_{\diadt{n}{j}}}                           
\Macro{zTjplus1}{\bm{z}^T_{\diadt{n}{{j+1}}}}                           
\Macro{Mxj}{M_{\bm{x}_{\diadt{n}{j}}}}                         
\Macro{Mxjplus1}{M_{\bm{x}_{\diadt{n}{{j+1}}}}}                         
\Macro{Mzj}{M_{\bm{z}_{\diadt{n}{j}}}}                         
\Macro{Mzjplus1}{M_{\bm{z}_{\diadt{n}{{j+1}}}}}
\Macro{SxO}{S\times\Omega}                           
\Macro{SigmaSxO}{\mathfrak{S}\times\mathfrak{O}}     
\Macro{Lp}{L^p_{}(X,\mathfrak{B}(X),m)}              
\Macro{L2}{L^2_{}(X,\mathfrak{B}(X),m)}
\Macro{LinftyS}{L^{\infty}_{}(S \times \Omega,\mathfrak{S}\times\mathfrak{O},\mu)}
\Macro{LinftyS*}{L^{\infty}_{}(S \times \Omega,\mathfrak{S}\times\mathfrak{O},\mu)^*}
\Macro{Mvof}{\mathfrak{M}_{v \circ f}^{}}              
\Macro{ImPhi}{\text{Im}\;\phi}                         
\Macro{KerPhi}{\text{Ker}\;\phi}                       
\Macro{GQKerPhi}{G\backslash\text{Ker}\phi}            
\Macro{CoMap}{\alpha\circ\delta\circ\gamma}             
\Macro{GrpA}{(G,*)}                                    
\Macro{GrpO}{(G^\prime,\diamond)}                      
\theoremstyle{plain}
\newtheorem{thm}{Theorem}[section]

\newtheorem{lem}[thm]{Lemma}
\newtheorem{prop}[thm]{Proposition}
\theoremstyle{plain}
\newtheorem{assumption}[thm]{Assumption}      
\theoremstyle{definition}

\theoremstyle{definition}

\theoremstyle{remark}
\newtheorem{rem}{Remark}[section]

\title{\textbf{Alpha Representation For Active Portfolio Management and High Frequency Trading In Seemingly Efficient Markets}}
\author{ Godfrey Cadogan \thanks{ Corresponding address: Institute for Innovation and Technology Management, Ted Rogers School of Management, Ryerson University, 575 Bay, Toronto, ON M5G 2C5; e-mail: \textcolor[rgb]{0.00,0.00,1.00}{\href{mailto:godfrey.cadogan@ryerson.ca}{godfrey.cadogan@ryerson.ca}}. I thank Bonnie K. Ray of IBM Watson Research Center, Program Chair, at the Joint Statistical Meeting 2011, Business and Economics Section  for her efforts in facilitating this work.  I am grateful to Steve Slezak, Victor K. Ng, and Gautam Kaul for introducing me to this topic, in its variegated forms, during Financial Economics seminars held at the University of Michigan many years ago. An expanded version of this paper with extensive literature review, and detailed proofs is available from the author upon request. Research support from the Institute for Innovation and Technology Management is gratefully acknowledged. All errors which may remain are my own.}}

\begin{document}

\maketitle
\begin{abstract}
We introduce  a trade strategy representation theorem for performance measurement and portable alpha in high frequency trading, by embedding a robust trading algorithm that describe portfolio manager market timing behavior,  in a canonical multifactor asset pricing model. First, we present a spectral test for market timing based on behavioral transformation of the hedge factors design matrix. Second, we find that the typical trade strategy process is a local martingale with a background driving Brownian bridge that mimics portfolio manager price reversal strategies. Third, we show that equilibrium asset pricing models like the CAPM exists on a set with P-measure zero. So that excess returns, i.e. positive alpha,  relative to a benchmark index is robust to no arbitrage pricing in turbulent capital markets. Fourth, the path properties of alpha are such that it is positive between suitably chosen stopping times for trading. Fifth, we demonstrate how, and why,  econometric tests of portfolio performance tend to under report positive alpha.
\end{abstract}

\section{Introduction\label{sec:Intro}}

The problem posed is one in which a portfolio manager ("PM") wants to increase portfolio alpha--the returns on her portfolio, over and above a benchmark or market portfolio. To do so [s]he alters the betas\footnote{See e.g., \cite{GrundyMalkiel1996} for viability of beta as a useful metric for covariance with benchmarks. \cite{Grinold1993} provides excellent exposition on the versatility of beta seperate from its use in the CAPM introduced by \cite{Sharpe1964}, \emph{inter alios}.} of the portfolio in anticipation of market movements by augmenting a benchmark model with hedge factors\footnote{See \cite{FamaFrench1996}; and \cite[pg.~276]{FungHsieh1997}}--which includes but is not limited to revising asset allocation or readjusting portfolio weights within an asset class. In other words, altered betas represent the managers dynamic trading strategy\footnote{Implicit in this assessment is the portfolio manager's response to good news or bad news accordingly--about assets in her portfolio--to exploit a so called leverage effect. See e.g., \cite{Black1976, BraunNelsonSunier1995}.}. Conceptually, the allocation of assets in the benchmark is ``fixed" but hedge factors are stochastic\footnote{See e.g., \cite[pg. ~10]{Jensen1967}.}--at least for so called ''portable alpha"\footnote{ See  \cite[pg.~78-79]{KungPohlman2004}. To wit, the portfolio may be ''market neutral' since benchmark and or market risk is hedged away..}.

This paper's contribution to behavioural finance, and the gargantuan market timing literature, stems from its reconciliation of active portfolio management with efficient markets when portfolio strategy or investment style is unobservable\footnote{See e.g., \cite{HenrikMerton1981, GrinblattTitman1989, FersonSchadt1996, MamayskySpiegelZhang2008, KacpercSialmZhang2008}.}. It employs asymptotic theory to identify an empirical portfolio alpha process with dynamic portfolio adjustments\footnote{See e.g., \cite{Urstadt2010}.} that reflect managerial strategy via martingale system equations that portend algorithmic trading. Additionally, it proves that the measurable sets for portfolio manager market timing ability are  much larger than those proffered in the extant literature which tests for timing ability via statistical significance of convex payoff structure(s)\footnote{See e.g. \cite{TreynorMazuy1966, TreynorBlack1973, Merton1981, BollenBusse2001}. Cf. \cite{GrinblattTitman1989, FersonSchadt1996}.}. Accordingly, we propose a new and simple test for market timing ability based on the spectral circle induced by a behavioural transformation of the hedge factor matrix.

\tab The paper proceeds as follows. In \autoref{sec:CanonicalAssetPricingModel} we formally introduce our model. Whereupon we summarize our representation theory result in Theorem \autoref{thm:GammaRepresentationTheorem}. Our spectral test for market timing is presented in Proposition \autoref{prop:SpectralTestForMarketTiming}. In \autoref{sec:ApplSingleFacModel} we apply our theory to the ubiquitous CAPM to provide analytics about Jensen's alpha. The main result there is Theorem \autoref{thm:AlphaPath} on the path process of positive alpha.

\section{The Canonical Linear Asset Pricing Model}\label{sec:CanonicalAssetPricingModel}

Let
\begin{equation}\label{eq:CanonicalAugCAPM}
   \bm{y}=X\pmb{\delta}+Z\pmb{\gamma}+\pmb{\epsilon}
\end{equation}
be the canonical hedge factor model, i.e., augmented capital asset pricing model (CAPM), for a portfolio comprised of: $X$--a matrix of returns from \emph{benchmark assets}\footnote{See \cite[pp. ~88-89]{GrinoldKahn2000} for explanation of benchmarking concept.}; and $Z$--a matrix of returns from \emph{hedge factors}\footnote{Arguably the most popular augmented CAPM-type benchmarking model is \cite{FamaFrench1993} (3-factor model includes; benchmark; small minus big stock returns (SMB); high minus low book to market stock returns (HML)). See \cite{Nohel2010} for a literature  review.} mimicking derivatives. The portfolio \emph{beta} is given by the row vector $\pmb{\beta}^T= (\pmb{\delta}^T\;\;\pmb{\gamma}^T)$ and $\pmb{\epsilon}$ is a column vector of idiosyncratic error terms\footnote{Column vectors are in bold print. The superscript T corresponds to transposition of a vector or matrix accordingly.}. The hedge factor strategy is embodied by $Z$. Thus, \emph{modulo} idiosyncratic error, our \emph{portfolio alpha} is given by
\begin{equation}\label{eq:PortfolioAlpha}
    \pmb{\alpha} = Z\pmb{\gamma}
\end{equation}
Whereupon $\pmb{\gamma}$ is hedge factor exposure sensitivity--it represents the \emph{trading strategy} of the portfolio manager\footnote{\cite[pg.~2]{JarrowProtter2010} identifies the constant intercept in a multifactor model as portfolio alpha. Our approach is tantamount to explaining that intercept with $Z$. See \cite[pg.~17-18]{AveryChevalierZeckhauser2011}.}. Similarly, $\pmb{\delta}$ is benchmark exposure sensitivity\footnote{See e.g. \cite[pg.~68]{TreynorBlack1973} for further interpretation and analytics.}. We would like to know what impact inclusion of $Z$ has on the model, including but not limited to its impact on returns $\bm{y}$\footnote{\cite[pg. ~5]{MacKinlayPastor1998} posited a similar parametrization except that they used a \cite{JamesStein1961} type estimation procedure to evaluate the impact of a missing factor on returns.}. For example, if inclusion of $Z$ has no impact, then $\pmb{\gamma}$ is statistically zero: our portfolio manager's choice of $Z$ is not generating \emph{alpha}. In the sequel our analyses are based on the following
\begin{assumption}[Filtered probability space]\label{assum:FilterProbSpace}
  $(\Omega,\mathcal{F},\mathbb{F},P)$. $\Omega$ is the sample space for states of nature; $\mathcal{F}$ is the $\sigma$-field of Borel measurable subsets of $\Omega$; $P$ is a probability measure defined on $\Omega$; and $\mathbb{F}=\{\mathcal{F}_s \subseteq \mathcal{F}_t \subseteq \mathcal{F};\;0\leq s < t < \infty\}$ is a filtration of sub $\sigma$-fields of $\mathcal{F}$.
\end{assumption}

\begin{assumption}
   $y:\Omega\rightarrow\mathbb{R}\backslash\{\{-\infty\},\{\infty\}\}$
\end{assumption}

\begin{assumption}
   $\text{P}-\lim_{T\rightarrow\infty}\frac{\sum^T_{t=1}x_tz_t}{T} =0$
\end{assumption}

\begin{assumption}\label{assum:HedgeFactorMatrix}
   The hedge factor matrix $Z(t,\omega)=(z_{ij}(t,\omega))\in  L^2(\Omega,\mathcal{F},P)$. Thus
   \begin{enumerate}
      {\indentitem \item[i.] $(z_{ij}(t,\omega)) is \mathcal{B}[0,\infty)\otimes \mathcal{F}$ measurable for the $\sigma$-field of Borel sets $\mathcal{B}$ generated on $[0,\infty)$.}
      {\indentitem \item[ii.]  $z_{ij}(t,\omega)$ is $\mathcal{F}_t$-adapted.}
      {\indentitem \item[iii.] $E[z^2_{ij}(t,\omega)]<\infty$.}
    \end{enumerate}
\end{assumption}

\begin{assumption}
   Markets are liquid so trades are executed at given prices.
\end{assumption}

\begin{assumption}
   Market microstructure effects are negligible.
\end{assumption}

\begin{assumption}
   $E[\epsilon] = 0,\quad E[\epsilon^2] < \infty$
\end{assumption}

\begin{assumption}
   $\pmb{\beta}$ is time varying.
\end{assumption}

To facilitate our asymptotic theory of portfolio alpha, we use a canonical dyadic partition of the unit interval $[0,1]$ starting at an arbitrary time $t=t_0$, on function space $C[0,1]$\footnote{Technical points involving Skorokhod space $D[0,1]$ are ignored here.}. In particular, $\prod^{(n)} = \{\diadt{n}{0} , \diadt{n}{1} ,\dotsc ,\diadt{n}{{m_n}}\}$ is a dyadic partition $\diadt{n}{j}=j.2^{-n}$ for $j=1 \ldots 2^n$. Let $y_{\diadt{n}{j}}$ be the augmented portfolio return at time $\diadt{n}{t}$; $\bm{x}^T_{\diadt{n}{j}}$ be the corresponding row vector of returns on the benchmark assets; and $\bm{z}^T_{\diadt{n}{j}}$ be the corresponding row vector of returns on hedge factors in the model. Let $\pmb{\delta}_{\diadt{n}{j}}$ and $\pmb{\gamma}_{\diadt{n}{j}}$ be the $\diadt{n}{j}$-th period coefficients, and
\begin{align}
   \Delta\pmb{\delta}_{\diadt{n}{{j+1}}} &= \pmb{\delta}_{\diadt{n}{{j+1}}}-\pmb{\delta}_{\diadt{n}{j}} , \quad
   \Delta\pmb{\gamma}_{\diadt{n}{{j+1}}} = \pmb{\gamma}_{\diadt{n}{{j+1}}}-\pmb{\gamma}_{\diadt{n}{j}}
\end{align}
be the corresponding change in model coefficients due to an additional observation\footnote{\cite[pp.~8-9]{FulkJordanSmith2010} used a similar parametrization to decompose portfolio returns into active and passive components.}.
\begin{assumption}
     $\Delta\pmb{\delta}_{\diadt{n}{{j+1}}}$ and $\Delta\pmb{\gamma}_{\diadt{n}{{j+1}}}$ are $\mathcal{F}_{{\diadt{n}{{j+1}}}^{-}}$-measurable.
\end{assumption}
To isolate the impact of the $j+1$-th period observation on the model we write
\begin{equation} \label{eq:AugBlockModel}
   \begin{pmatrix} 
      \bm{y}_{\diadt{n}{j}}\\ \hdots \\ y_{\diadt{n}{{j+1}}} 
   \end{pmatrix}
   =                             
   \begin{pmatrix} 
      X_{\diadt{n}{j}} &\vdots &Z_{\diadt{n}{j}} \\
      \hdotsfor{3} \\
      \pmb{x}^T_{\diadt{n}{{j+1}}} &\vdots &\pmb{z}^T_{\diadt{n}{{j+1}}}
   \end{pmatrix}
   \begin{pmatrix} 
      \hat{\pmb{\delta}}_{\diadt{n}{j+1}} \\
      \hdotsfor{1} \\
      \hat{\pmb{\gamma}}_{\diadt{n}{j+1}}
   \end{pmatrix}
   +                             
   \begin{pmatrix} 
      \bm{e}_{\diadt{n}{j}} \\
      \hdotsfor{1} \\
      e_{\diadt{n}{{j+1}}}
   \end{pmatrix}
\end{equation}
where $e$ is the sample estimate of $\epsilon$. In which case we get the linear relation
\begin{align} 
   \bm{y}_{\diadt{n}{j}} &=X\vechat{\delta}_{\diadt{n}{{j+1}}}+Z\vechat{\gamma}_{\diadt{n}{{j+1}}}+\bm{e}_{\diadt{n}{{j+1}}} \label{eq:AugExpostErr}\\
   y_{\diadt{n}{{j+1}}} &= \bm{x}^T_{\diadt{n}{{j+1}}}\hat{\delta}_{\diadt{n}{{j+1}}}+\bm{z}^T_{\diadt{n}{{j+1}}}\hat{\gamma}_{\diadt{n}{{j+1}}}+e_{\diadt{n}{{j+1}}} \label{eq:AugForecast}
\end{align}
where $\bm{x}^T_{\diadt{n}{{j+1}}}$ and $\bm{z}^T_{\diadt{n}{{j+1}}}$ are \emph{row vectors}. So that if there are $m$ assets in the benchmark portfolio, and $p$ hedge factors/assets, then $X_{\diadt{n}{j}}=[\bm{x}_{\diadt{n}{1}}\dotsc\bm{x}_{\diadt{n}{j}}]$ is a $j\times m$ matrix, and $Z_{\diadt{n}{j}}=[\bm{z}_{\diadt{n}{1}}\dotsc\bm{z}_{\diadt{n}{j}}]$ is a $j\times p$ matrix. An additional observation appends a row vector to each matrix accordingly\footnote{In the sequel we suppress the time subscript for the $X_{\diadt{n}{j}}$ and $Z_{\diadt{n}{j}}$ matrices, and write $X$ annd $Z$  for notational convenience. However, we reserve the right to invoke the time subscript as necessary..}. So that $Z$ is really a progressively measurable $j \times p$ matrix process for $j=0,1,\dotsc,2^n$.

\subsection{Behavioural Heuristics On Altering Beta}\label{subsec:HeuristicsOnAlterBeta}

\tab Technically,  $\bm{z}^T_{\diadt{n}{{j+1}}}(\omega)$ is not $\mathcal{F}_{\diadt{n}{j}}$-adapted. That is, it cannot be determined solely from information in $\mathcal{F}_{\diadt{n}{j}}$. The portfolio manager must be ``clairvoyant" and find some algebraic number\footnote{See e.g., \cite[pg.~88]{Clark1971} for definition of algebraic number and related concepts introduced here.} in $\mathcal{F}_{\diadt{n}{{j+1}}}$. The gist of \cite{Cadogan2011} is that implied volatility ($\sigma$) from options prices is such a ``clairvoyant" algebraic number\footnote{See also, \cite{BakshiPanayotovSkoulakis2010} who showed that forward looking volatility, i.e. an algebraic number, from options market have predictive power for asset returns. At a more technical level, \cite[pg.~10]{Myneni1992} used martingale theory from \cite[pg.~135, 74(b)]{DellacherieMeyer1982} to advocate for the existence of \emph{dual predictable projection} of processes with integrable variation. To wit, if $Z$ is convex--as hypothesized, then it satisfies the dual predictable projection criterion.}. Therefore, for some closed class of polynomials $\mathcal{P}$, and polynomials $g,\;h\in\mathcal{P}$, the hedge factor(s) $\bm{z}^T_{\diadt{n}{{j+1}}}$ can be expressed as a polynomial $g(\sigma)$ for $\sigma\in\mathcal{F}_{\diadt{n}{{j+1}}}$ with coefficients drawn from $\mathcal{F}_{\diadt{n}{j}}$. In other words, returns forecast must be based on forward $[g(\sigma)]$ and backward $[h(y_{\diadt{n}{1}},\;y_{\diadt{n}{2}},\;\dotsc,\;y_{\diadt{n}{j}})]$ looking variables based on derivative pricing.   So that
\begin{align}
      y_{\diadt{n}{{j+1}}}  &= g(\sigma) + h(y_{\diadt{n}{1}},\;y_{\diadt{n}{2}},\;\dotsc,\;y_{\diadt{n}{j}})+\epsilon_{\diadt{n}{{j+1}}}
\end{align}
In which case for $\pmb{x}_{\diadt{n}{j}}$ fixed in \ref{eq:AugForecast}, $\pmb{z}_{\diadt{n}{{j+1}}}=g(\sigma)$ is the contribution of new information to returns, $y_{\diadt{n}{{j+1}}}$, after parameter updates\footnote{See e.g., \cite{AdmatiPfleiderer1988} for evolution of trade patterns and information flows.}. In a nutshell, $\pmb{z}_{\diadt{n}{{j+1}}}$ is predictable\footnote{See \cite[pg.~21]{KaratzasShreve1991}.}; thus paving the way for its use in martingale transform equations. These results are summarized in the following
\begin{lem}[Predictable hedge factors]\label{lem:PredictableHedgeFactors}~\\
   Let $\bm{z}^T_{\diadt{n}{{j+1}}}$ be a vector of returns isomorphic to the terminal payoff of a contingent claim, and $\sigma$ be an algebraic number in $\mathcal{F}_{\diadt{n}{{j+1}}}$. Let $\mathcal{P}$ be the class of closed polynomials with coefficients in $\mathcal{F}_{\diadt{n}{j}}$. Then $\bm{z}^T_{\diadt{n}{{j+1}}}=g(\sigma)$ is predictable.
   \begin{flushright}
      $\Box$
   \end{flushright}
\end{lem}
\begin{rem}
   \cite{Kassouf1969} provides empirical support for this lemma.
\end{rem}

The dispositive question here is how to alter the portfolio`s beta, i.e., forecast $\delta_{\diadt{n}{{j+1}}}$ and $\gamma_{\diadt{n}{{j+1}}}$, to maximize next period`s returns. The vector of returns is given by
\begin{align}
   \pmb{y}_{\diadt{n}{{j+1}}} &= [\bm{y}^T_{\diadt{n}{j}} : y^T_{\diadt{n}{{j+1}}}]^T.\;\;\text{Whereupon}\\
   \bm{y}_{\diadt{n}{j}} &=X\pmb{\delta}_{\diadt{n}{{j+1}}}+Z\pmb{\gamma}_{\diadt{n}{{j+1}}}+\pmb{\epsilon}_{\diadt{n}{j}}\label{eq:ExPostPortfolioRet}\\
                         &=X\pmb{\delta}_{\diadt{n}{j}} + \overbrace{X\Delta\pmb{\delta}_{\diadt{n}{{j+1}}}+Z\pmb{\gamma}_{\diadt{n}{{j+1}}}+\pmb{\epsilon}_{\diadt{n}{j}}}^{\text{\emph{ex post} tracking error}},\;\;\text{and}\label{eq:TrackingError}\\
   y_{\diadt{n}{{j+1}}} &= \bm{x}^T_{\diadt{n}{{j+1}}}\pmb{\delta}_{\diadt{n}{{j+1}}} + \overbrace{\bm{z}^T_{\diadt{n}{{j+1}}}\pmb{\gamma}_{\diadt{n}{{j+1}}}+\epsilon_{\diadt{n}{{j+1}}}}^{\text{\emph{ex ante} tracking error}}\label{eq:ExAntePortfolioRet}
\end{align}
Ideally, the portfolio manager would like tracking error to be zero as she tries to replicate the benchmark and or index in \ref{eq:TrackingError}. See e.g., \cite[pp.~676-677]{EltonGruberBrownGoetz2003}. See also, \cite[pg. ~49]{GrinoldKahn2000} who define ``tracking error" as ``how well the portfolio can track the benchmark". It is the ``active returns" on the portfolio. This is tantamount to imposing the following behavioral restrictions on the \emph{ex post} tracking error equation
\begin{align}
   X\Delta\pmb{\delta}_{\diadt{n}{{j+1}}}+Z\pmb{\gamma}_{\diadt{n}{{j+1}}}+\pmb{\epsilon}_{\diadt{n}{j}} &= 0\\
   \intertext{If the proportion of assets in the benchmark is fixed--technically this is a ''portable alpha" strategy, then}
   \Delta\pmb{\delta}_{\diadt{n}{{j+1}}} &= 0,\;\;\text{and}\\
   \hat{\pmb{\gamma}}^{\text{res}}_{\diadt{n}{{j+1}}} &= -\bigl(Z^TZ\bigr)^{-1}Z^T\pmb{\epsilon}_{\diadt{n}{j}}\label{eq:RestrictedGamma}\\
   \intertext{Thus, hedge factor exposure sensitivity plainly depends on, \emph{inter alia}, the behavior of $\pmb{\epsilon}_{\diadt{n}{j}}$. Consistent with our augmented model, define the projection matrices, see e.g., \cite[pp.~149-150]{Greene2003}}
   P_X &= X(X^TX)^{-1}X^T\\
   P_Z &= Z(Z^TZ)^{-1}Z^T\\
   M_X &= I-P_X\\
   M_Z &= I-P_Z\\
   \intertext{So that assuming that $X$ and $Z$ are uncorrelated with $\pmb{\epsilon}$ we have the unrestricted estimate, see \cite[pp.~121-122]{ChristophFersonGlassman1998},}
   \hat{\pmb{\gamma}}^{\text{unres}}_{\diadt{n}{{j+1}}} &= \bigl(Z^TM_XZ\bigr)^{-1}Z^TM_X\bm{y}_{\diadt{n}{j}}\label{eq:UnrestrictedGamma}\\
   \intertext{Our portfolio manager has superior market timing ability, see \cite[pg.~436]{FersonSchadt1996}, if}
   \bm{z}^T_{\diadt{n}{{j+1}}}\pmb{\gamma}^{\text{res}}_{\diadt{n}{{j+1}}}+\epsilon_{\diadt{n}{{j+1}}} & \geq 0\label{eq:SuperiorTimingAbility}\\
   \intertext{So we can rewrite \ref{eq:ExAntePortfolioRet} as follows}
   y_{\diadt{n}{{j+1}}} &= \bm{x}^T_{\diadt{n}{{j+1}}}\pmb{\delta}_{\diadt{n}{{j+1}}} + \max\{0,\;{\bm{z}^T_{\diadt{n}{{j+1}}}\pmb{\gamma}_{\diadt{n}{{j+1}}}+\epsilon_{\diadt{n}{{j+1}}}}\}\label{eq:AugContingentClaimForm}\\
   \intertext{Whereupon substitution of $\hat{\pmb{\gamma}}^{\text{res}}_{\diadt{n}{{j+1}}}$ from \ref{eq:RestrictedGamma} in \ref{eq:SuperiorTimingAbility} yields}
   \epsilon_{\diadt{n}{{j+1}}} &\geq \bm{z}^T_{\diadt{n}{{j+1}}}\bigl(Z^TZ\bigr)^{-1}Z^T\pmb{\epsilon}_{\diadt{n}{j}}\label{eq:PM_SubmartingaleBeliefs}
\end{align}
The functional form in \ref{eq:AugContingentClaimForm} is equivalent to \cite[pp.~365-366,~368-369]{Merton1981} formulation of isomorphism between the pattern of returns from market timing and returns on an option strategy\footnote{See also, \cite{GlostenJagannathan1994} and \cite[pg.~68]{AgrawalNaik2004} for  extension(s).}.  Intuitively, our parametrization implies that the benchmark is perfectly tracked. Thus, any mispricing in the model stems from the PM performance in selecting hedge factors or contingent claims. In any event, \ref{eq:PM_SubmartingaleBeliefs} suggests that if our portfolio manager is bullish, i.e. she believes that the returns process is a semi-martingale that is favorable to her, see e.g., \cite[pg.~299]{Doob1953}, then 
\begin{equation}
   \bm{z}^T_{\diadt{n}{{j+1}}}\bigl(Z^TZ\bigr)^{-1}Z^T\pmb{\epsilon}_{\diadt{n}{j}} \geq \epsilon_{\diadt{n}{j}}\label{eq:PM_SubmartingaleBeliefs2}
\end{equation}
Equations \ref{eq:PM_SubmartingaleBeliefs} and \ref{eq:PM_SubmartingaleBeliefs2} gives rise to the following
\begin{thm}[Market Timing Theorem]\label{thm:MarketTiming}
   Let $Z$ be a matrix of hedge factors at time $\diadt{n}{j}$ and $\bm{z}^T_{\diadt{n}{{j+1}}}$ be an additional row vector of future observations, i.e., derivative prices of the hedge factors. Furthermore, suppose that $\pmb{\epsilon}_{\diadt{n}{j}}$ is a vector of portfolio manager forecast errors, and $\epsilon_{\diadt{n}{{j+1}}}$ is forecast error at time $\diadt{n}{{j+1}}$. Assume that $Z$ and $\epsilon_{\cdot}$ are uncorrelated, and that $\epsilon_{\cdot}\sim (0,1)$. Then our portfolio manager has market timing ability iff
   \begin{equation}
      \sup_{0\leq j\leq 2^n}\|E[\bm{z}^T_{\diadt{n}{{j+1}}}]\bigl(Z^TZ\bigr)^{-1}Z^T\|^2 \geq 2^{-n}
   \end{equation}
   \begin{flushright}
      $\square$
   \end{flushright}
\end{thm}
\begin{rem}
   The theorem essentially implies that as trading frequency increases, i.e. $n\uparrow\infty$, our portfolio manager will have timing ability for any previsible process $\{\bm{z}_t,\mathcal{F}_t;\;t\geq 0\}$. This is  the \emph{sui generis} of market timing. It constitutes a mathematical proof of \cite{ChanceHemler2001} empirical results which found that the same portfolio managers who seemingly lacked timing ability at low frequency were found to have timing ability at high frequency.
\end{rem}
\subsection{The Martingale System Equation For Market Timing}\label{subsec:MartingaleSystemForMarketTiming}
This section develops the martingale representation theory. See \cite[pp. ~363-365]{Dudley2004} and \cite[Chapter~5]{Breiman1968}  for excellent summary of martingales. Let
\begin{equation}
   u_j(\omega)=
   \begin{cases}
      1 & \text{if $ \bm{z}^T_{\diadt{n}{{j+1}}}\pmb{\gamma}^{\text{res}}_{\diadt{n}{{j+1}}}+\epsilon_{\diadt{n}{{j+1}}}(\omega) > 0$}\\
      0 & \text{if $ \bm{z}^T_{\diadt{n}{{j+1}}}\pmb{\gamma}^{\text{res}}_{\diadt{n}{{j+1}}}+\epsilon_{\diadt{n}{{j+1}}}(\omega) \leq 0$}
   \end{cases}
\end{equation}
and define
\begin{align}
    d_{k+1}&=y_{\diadt{n}{{k+1}}}-\pmb{x}_{\diadt{n}{k}}^T\pmb{\delta}_{\diadt{n}{{k+1}}}\\
    \intertext{So that the equation}
    \bar{d}_n &= d_1 + \sum^{n-1}_{j=1}u_j(\omega)d_{j+1}\label{eq:MartingaleSysEq}
\end{align}
represents the excess returns from the given portfolio strategy. This is the martingale system equation referred to in \cite[pg.~295]{Snell1952}. In the context of our model it represents the portfolio manager data mining algorithm which propels her high frequency trades. The specific strategy in place can be seen from rewriting the equation as
\begin{align}
   \bar{d}_n &= d_1 + \sum^{2^n-1}_{k=1}\bigl(y_{\diadt{n}{{k+1}}}-\pmb{x}_{\diadt{n}{k}}^T\pmb{\delta}_{\diadt{n}{{k+1}}} \bigr)^+\\
             &= d_1 + \sum^{2^n-1}_{k=1}\bigl(\bm{z}_{\diadt{n}{{k+1}}}\pmb{\gamma}_{\diadt{n}{{k+1}}}+\epsilon_{\diadt{n}{{k+1}}}\bigr)^+
\end{align}
where the summand is tantamount to a call option on the benchmark\footnote{In our case, the call option is on some hedge factor(s) that are uncorrelated with the benchmark \emph{per se}. Arguably, the benchmark constitutes the ''microforecast" while the hedge factor(s) comprise the ''macroforecast" or market timing ability. See \cite{Fama1972}.}, as indicated by \cite{Merton1981, HenrikMerton1981}. See also, \cite[pg.~77]{Henricksson1984}. According to \cite[Thm.~2.1,~pg.~295]{Snell1952} the sequence $\{d_n,\mathcal{F}_n;n\geq 1\}$ is a semimartingale in which $E[\bar{d}_n|\;\mathcal{F}_1]\leq E[d_n|\;\mathcal{F}_1]$. For our purposes it implies that in an efficient market, in the long run, the portfolio manager should be no better off by ``judicious" selection of favorable $\bar{d}_n$ transforms, i.e., option(s) strategies. These artifacts are summarized in a slightly modified version of Snell`s Theorem as follows:
\begin{prop}[Snell`s Theorem]\label{prop:SnellTheorem}\cite[Thm.~2.1,~pg.~295]{Snell1952}.\\
    Let $(\Omega,P,\mathcal{F})$ be a probability space; $D=\{d_k,\mathcal{F}_k;k\geq 1\}$ be a martingale; and $\{u_k(\omega);\;k\geq 1\}$ be a sequence of $\mathcal{F}_k$-measurable random variables. Define
   \begin{equation*}
      \bar{d}_k = d_1 + \sum^{k-1}_{j=1}u_j(\omega)d_{k+1}
   \end{equation*}
   If $E[|\bar{d}_k|]<\infty$ for all $k$, then $\bar{D}=\{\bar{d}_k,\;\mathcal{F}_k;\;k\geq 1\}$ is a martingale, and the $u_k$`s are nonnegative, then $\bar{D}$ is a submartingale. If the $u_k$`s are binary random variables taking the values $0$ or $1$, then we have
   \begin{equation*}
      E[\bar{d}_k|\mathcal{F}_1]\leq E[d_k|\mathcal{F}_1]
   \end{equation*}
   with probability 1.
   \begin{flushright}
      $\Box$
   \end{flushright}
\end{prop}
\begin{proof}
   See \cite{Snell1952}.
\end{proof}

\subsection{Trade strategy in continuous time, and statistical test for market timing}\label{subsec:TradeStrategyStat_Test}
\tab In this subsection we state some of our main results--most with referenced proofs. Equating \ref{eq:RestrictedGamma} and \ref{eq:UnrestrictedGamma} gives rise to the following
\begin{prop}[Spectral test for market timing]\label{prop:SpectralTestForMarketTiming}~\\
    Let $Z$ be a $j\times p$ matrix of hedge factors, $X$ be a $j\times m$ matrix of benchmark assets, and $ P_X = X(X^TX)^{-1}X^T$ be the projection matrix on $X$-space. Define $A = Z^T(2I-P_X)Z$ where $I$ is the identity matrix. Let $\lambda_k(A)$ be the $k$-th eigenvalue of A. Let $\eta > 0$ be a suitably chosen number. Then our portfolio manager has timing ability if
    \begin{equation*}
       \max_{1\leq k\leq p}|\lambda_k(A)| > \eta
    \end{equation*}
    Moreover, this is tantamount to the statistical test:
    \begin{align*}
        H_0:\;\max_{1\leq k\leq p}\;\lambda_k(A) &\leq \eta \quad \emph{versus} \quad  H_a:\; H_0 \;\text{is not true}
    \end{align*}
    \begin{flushright}
       $\Box$
    \end{flushright}
\end{prop}
\begin{rem}
   The exact statistical distribution for $\max_{1\leq k\leq p}\;\lambda_k(A)$ is a fairly complex looking function given in \cite{ErtenZandonaReigber2009}. Moreover, in practice it is possible for $\lambda$ to be negative based on numerical routines.
\end{rem}
\begin{rem}
    \cite[Cor.~6.1,~pg.~200]{HansenScheinkman2009} derived a principal eigenvalue result by applying semigroup theory to a stochastic discount factor assumed to follow a Markov process.
\end{rem}
Nonetheless, to computer the power of our spectral test we proffer the following
\begin{thm}[Power of spectral test for market timing]\label{thm:PowerOfSpectraltest}
    If $\ell_1=\max_{1\leq k\leq p}\;\lambda_k(A)$ is the largest latent root of $A$, and $A=H^TH$, where $H\sim N(0,I_n\otimes\Sigma)$ and $W_p(n,\Sigma)$ is a Wishart distribution with n-degrees of freedom and dimension $p$, $A\sim W_p(n,\Sigma)$, then the distribution function for $\ell_1$ can be expressed as
    \begin{equation}\label{eq:DistFuncLambdaMax}
       P_{\Sigma}(\ell < \eta) = \frac{\Gamma_m[\tfrac{1}{2} (m+1)]}{\Gamma_m[\tfrac{1}{2} (n+m+1)]}\;\det(\tfrac{1}{2} n\eta \Sigma^{-1})^{\tfrac{n}{2}}
                               {}_{1}F_1(\tfrac{n}{2};\;\tfrac{1}{2} (n+m+1);\;-\tfrac{1}{2} n\eta \Sigma^{-1})
    \end{equation}
    where ${}_{1}F_1(\cdot)$ is a hypergeometric function such that
    \begin{equation*}
        {}_{p}F_q (a_1,\dotsc,a_p;\;\;b_1,\dotsc,b_q;\;\;z) = \sum^{\infty}_{k=0}\frac{(a_1)_k\dotsc(a_p)_k}{(b_1)_k\dotsc(b_q)_k}\frac{z^k}{k!}
    \end{equation*}
    where $(a)_k=a(a+1)\dotsc(a+k-1)$, and $\Gamma_m (\cdot)$ is a multivariate gamma function.
\end{thm}

\begin{proof}
   See \cite[pg.~421,~Cor.~9.7.2]{Muirhead2005}.
\end{proof}

\begin{rem}
   The multivariate gamma function $\Gamma_m(\cdot)$ is defined in \cite[pg.~61]{Muirhead2005}.
\end{rem}

\begin{thm}[Subordinated Brownian motion]\label{thm:SubordinateBrownMot}
Let $\epsilon_{\diadt{n}{j}}$ be independent and identically distributed with $E[\epsilon_{\diadt{n}{j}}]=0$ and $E[\epsilon_{\diadt{n}{j}}^2]=\sigma^2<\infty$, for $j=1,2,\dotsc,2^n$. Let $S_N = \Sigma^N_{j=1}\epsilon_{\diadt{n}{j}}$ and for $\diadt{n}{j} \leq t < \diadt{n}{{j+1}}$ define
   \begin{equation*}
      \epsilon^{(n)}_t = \frac{1}{\sqrt{n}}[S_{[nt]} + (nt-[nt])\epsilon_{[nt]+1}
   \end{equation*}
Then $\epsilon^{(n)}_{t+2^{-n}}-\epsilon^{(n)}_t$ is a subordinated Brownian motion for some strictly monotone function $c(\cdot)$. In particular, $\epsilon^{(n)}_{t+2^{-n}}(\omega)-\epsilon^{(n)}_t(\omega) \sim B_{c(t)}(\omega)$ on the probability space $(\Omega,\;\mathcal{F},\;P)$.
\end{thm}
\begin{proof}
   See \autoref{apx:ProofSubordinateBrowninanMotion}.
\end{proof}

\begin{thm}[Trading strategy representation. \cite{Cadogan2011d}]\label{thm:GammaRepresentationTheorem}~\\
   Let $(\Omega,\mathcal{F},\mathbb{F},P)$ be a filtered probability space, and $Z=\{Z_s,\mathcal{F}_s;\; 0\leq s < \infty\}$ be a hedge factor matrix process on the augmented filtration $\mathbb{F}$. Furthermore, let $a^{(i,k)}(Z_s)$ be the $(i,k)$-th element in the expansion of the transformation matrix $(Z_s^TZ_s)^{-1}Z^T_s$, and $B=\{B(s),\mathcal{F}_s;\;s\geq 0\}$ be Brownian motion adapted to $\mathbb{F}$ such that $B(0)=x$.  Let $\gamma^{(i)\Pi^{(n)}}(t,\omega) = -\sum^j_{k=1}a^{(i,k)}(Z_{t^*_k})\epsilon_{\diadt{n}{k}}\chi_{[\diadt{n}{{j-1}},\diadt{n}{j})}(t),\;\;\diadt{n}{{j-1}}<t^*_k<\diadt{n}{j}$, with respect to partition $\Pi^{(n)}$ and characteristic function $\chi_{[\diadt{n}{{j-1}},\diadt{n}{j})}(t)$. Assuming that $B$ is the background driving Brownian motion for high frequency trading, the limiting hedge factor sensitivity process, i.e. trading strategy, $\gamma = \{\gamma_s,\mathcal{F}_s;0\leq s < \infty\}$ generated by portfolio manager market timing for Brownian motion starting at the point $x\geq 0$ has representation
   \begin{equation*}
      d\gamma^{(i)}(t,\omega) = \sum^j_{k=1}a^{(i,k)}(Z_t)\biggl[\frac{x}{1-t}\biggr]dt  - \sum^j_{k=1}a^{(i,k)}(Z_t)dB(t,\omega),\;x\geq 0
   \end{equation*}
   for the $i$-th hedge factor $i=1,\dotsc,p$, and $0\leq t\leq 1$.
\end{thm}
\begin{proof}
   Apply Theorem \autoref{thm:SubordinateBrownMot} to $\lim_{n\rightarrow\infty}\gamma^{(i)\Pi^{(n)}}(t,\omega)$. See \cite[Thm.~4.6]{Cadogan2011d}.
\end{proof}

\section{Application: Dynamic alpha in a single factor model}\label{sec:ApplSingleFacModel}
We employ our trade strategy representation theorem, to shed light on the behavior of portfolio \emph{alpha} in a single factor model like CAPM, where there is no hedge factor. In particular, let $\mathbb{1}_{\{n\}}$ be a $n\times 1$ vector, and
\begin{align}
   Z &=\mathbb{1}_{\{n\}}\\
   \intertext{So that}
   (Z^TZ)^{-1}Z^T &= n^{-1}\mathbb{1}_{\{n\}}^T,\;\;\text{and}\;\;a^{(1,k)}(Z_s) = n^{-1},\;k=1,\dotsc,n\\
   \intertext{Substitution of these values in \ref{eq:PortfolioAlpha} and Theorem \autoref{thm:GammaRepresentationTheorem} gives us}
   \alpha^{(1)}(t) &= \gamma^{(1)}(t)\\
   -d\alpha^{(1)}(t) &= -\frac{x}{1-t}dt + dB(t)\\
   \intertext{That is the equation of a Brownian bridge starting at $B(0)=x$ on the interval $[0,1]$. See \cite[pg.~268]{KarlinTaylor1981}. So that}
   d\alpha^{(1)}(t) &= -dB^{br}(t)\\
   \alpha^{(1)}(t) &= B^{br}(0)-B^{br}(t)\label{eq:GammaEstimate}\\
   \intertext{The Brownian bridge feature suggests that portfolio managers open and close their net positions at zero, and take profits (or losses) in between. See e.g., \cite{Urstadt2010}. And the negative sign implies that our portfolio manager is engaged in a price reversal strategy. See e.g., \cite[pp.~14-15]{Brogaard2010}. So that $B^{br}(t) <0 \Rightarrow \alpha^{(1)}(t)>0$. According to Girsanov's formula in \cite[pg.~162]{Oksendal2003}, we have an equivalent probability measure $Q$ based on the martingale transform}
   M(t,\omega) &= \exp\biggl(\int^t_0\frac{x}{1-s}dB(s,\omega)-\int^t_0\biggl(\frac{x}{1-s}\biggr)^2 ds\biggr)\\
   dQ(\omega) &= M(T,\omega)dP(\omega),\quad 0\leq t\leq T\leq 1\\
   \intertext{Thus, we have the $Q$-Brownian motion, i.e. Brownian bridge}
   \hat{B}(t) &= -\int^t_0 \frac{x}{1-s}ds + B(t),\quad \text{and}\label{eq:Q_BrownianMotion}\\
   d\alpha^{(1)}(t) &= -d\hat{B}(t)=-dB^{br}(t)\\
   \intertext{In other words, $\alpha^{(1)}$ is a $Q$-Brownian motion, i.e. Brownian bridge, that reverts to the origin starting at $x$. We note that for idiosyncratic risk $\epsilon(t)$, the CAPM holds if $\alpha^{(1)}(t)+\epsilon(t)=0$, and}
   \frac{d\alpha^{(1)}(t)}{dt} + \frac{\epsilon(t)}{dt} &= -\frac{dB^{br}(t)}{dt} +\frac{\epsilon(t)}{dt}=0\\
   \intertext{Hence the ''residual(s)" $\epsilon(t)$, associated with alpha, have an approximately skewed U-shape pattern if  $B^{br}(t) \leq  0$. \cite[pg.~358]{KaratzasShreve1991} also provide further analytics which show that on $[0,1]$ we can write the portfolio alpha process in mean reverting form as}
   d\alpha^{(1)}(t) &= \frac{1-\alpha^{(1)}(t)}{1-t}dt+dB(t);\quad 0\leq t\leq 1,\;\;\alpha^{(1)}(0)=0\label{eq:GammaRepresentationAlternative}\\
   M(t) &= \int^t_0\frac{dB(s)}{1-s}\\
   T(s) &= \inf\{t| <M>_t > s\}\\
   G(t) &= B_{<M>_{T(t)}}\\
   \intertext{Thus, under Dambis-Dubins-Schwarz criteria, alpha is a time changed martingale--in this case Brownian motion. In the absence of a hedge factor, the single factor or benchmark, is perfectly tracked if}
   \alpha^{(1)}(t) &= 0,\quad \hat{B}(t,\omega) = x
\end{align}
 The foregoing gives rise to the following
\begin{thm}[Positive CAPM alpha excursion]\label{thm:AlphaPath}~\\
   Let $\alpha^{(1)+}(t)$ be the $Q$-Brownian motion excursion path of CAPM alpha at time $t$ in \ref{eq:Q_BrownianMotion},  and $B(t)$ be standard one-dimensional Brownian motion. Let $\tau^\alpha_+(t)$ be the first zero of $B$ after and $\tau^\alpha_-(t)$ be the first zero of $B$ before $t=1$. So that
   \begin{align}
      \tau^\alpha_+(t) &= \inf \{t>1| B(t)=0\}\\
      \tau^\alpha_-(t) &= \sup \{t<1| B(t)=0\}\\
      \intertext{Then}
      \alpha^{(1)+}(t) &= \frac{|B(t\tau^\alpha_+ + (1-t)\tau^\alpha_-(t)) |}{\sqrt{\tau^\alpha_+(t)-\tau^\alpha_-(t)}}
   \end{align}
\begin{flushright}
   $\Box$
\end{flushright}
\end{thm}
\begin{proof}
   See \cite{Vervaat1979}.
\end{proof}

\tab Thus, the path properties of portfolio alpha can be identified and excess returns can be computed for suitably chosen stopping times. The property $\hat{B}(t,\omega) = x$ reduces the problem to one of local time [of a Brownian bridge] at $x$. We can think of $x$ as a hurdle rate such as transaction costs that the manager must attain to break even. The probability associated with the CAPM alpha level set $\mathfrak{B}=\{\omega|\;\hat{B}(t,\omega)=x\}$ is zero. However, even though that set has P-measure zero, its local time exists. Perhaps more important, the perfectly hedged portfolio problem, i.e. the CAPM problem, reduces to one of stochastic optimal control--guiding $\alpha^{(1)}$ to a goal of $0$ by keeping it as close to $0$ as possible. This problem, and related ones, were solved by \cite{BenesSheppWitsen1980} and in \cite[Chapter~6.2]{KaratzasShreve1991}.

\subsection{On spurious econoometric tests for alpha}
According to \cite[pg.~269]{KarlinTaylor1981} the expected value of alpha starting at $x$, and its variance is given by
\begin{align}
   E[\alpha^{(1)}|\;\tilde{B}(t)] &= -\frac{x}{1-t},\quad \sigma^2_{\alpha^{(1)}} = 1\label{eq:AverageAlpha}\\
   \intertext{Let$\{\alpha_1^{(1)},\dotsc,\alpha_N^{(1)}\}$ be a sample of alphas for $N$-funds. Furthermore, assume that the fund alphas are pairwise correlated with correlation coefficient $\rho_{ij}$. Cf. \cite[pp.~17-19]{AveryChevalierZeckhauser2011}. So that}
   \alpha_i^{(1)} &= \rho_{ij}\alpha_j^{(1)},\quad |\rho_{ij}|<1\\
   \intertext{The sample mean and variance of the funds are given by}
   \bar{\alpha}_N^{(1)} &= \frac{1}{N}(\alpha_1^{(1)} +\dotsc +\alpha_N^{(1)}\}\\
   \sigma^2_{\bar{\alpha}^{(1)}} &= \frac{1}{N^2}(\sigma^2_{\alpha_1^{(1)}}+\dotsc +\sigma^2_{\alpha_N^{(1)}}+2\sum_{i\neq j}\text{cov}(\alpha_i^{(1)},\alpha_j^{(1)}))\\
   |\sum_{i\neq j}\text{cov}(\alpha_i^{(1)},\alpha_j^{(1)})|&\leq \sum_{i\neq j}|\text{cov}(\alpha_i^{(1)},\alpha_j^{(1)})|\leq \sum_{i\neq j}|\rho_{ij}|\leq \binom{N}{2}
   \intertext{In that milieu, a $t$-test for the hypothesis $H_0:\alpha^{(1)}=0$ has test statistic}
   t_{\bar{\alpha}_N^{(1)}} &= \frac{\bar{\alpha}_N^{(1)}}{\sigma_{\bar{\alpha}_N^{(1)}}} = \frac{\bar{\alpha}_N^{(1)}}{\sqrt{\frac{1}{N}+\frac{2}{N^2}\sum_{i\neq j}\text{cov}(\alpha_i^{(1)},\alpha_j^{(1)})}}\\
   &\geq \frac{\bar{\alpha}_N^{(1)}}{\sqrt{\frac{1}{N} + 1}},\;\text{for sufficiently large $N$}\\
   \lim_{N\rightarrow\infty}t_{\bar{\alpha}_N^{(1)}} &= Z_{\bar{\alpha}_\infty^{(1)}}\geq \bar{\alpha}_\infty^{(1)}, \;\text{where $Z_{\bar{\alpha}_\infty^{(1)}}$ is a standard normal r.v.}
\end{align}
Ergodic theory\footnote{See e.g., \cite[pg.~127]{GikhmanSkorokhod1969}} tells us that the limiting value of the test statistic $\bar{\alpha}_\infty^{(1)}$ is a Brownian bridge\footnote{In this heuristic example, we ignored issues arising from seemingly unrelated regressions or confidence sets.}.  Moreover, according to \ref{eq:AverageAlpha}, it tends to be negative valued. Thus, an analyst could easily conclude that the sampled funds do not generate positive alpha\footnote{\cite[pg.~1308]{Phillips1998} noted that the time trend component in a Brownian bridge--in our case alpha--contributes to spurious regression.Also,  \cite[pg.~1398]{FersonSarkSimin2003} cautioned about seemingly significant $t$-ratios derived from spurious regressions.}. Yet, we know from the path properties in Theorem \autoref{thm:AlphaPath} that there are stopping times for which the funds do generate positive alpha. So contrary to \cite[pg.~19]{Jarrow2010}  false positive alpha postulate, our theory indicates that there is a false negative alpha puzzle.

\section{Appendix}\label{apx:Appendix}
\appendix
\singlespace

\section{Proof of subordinated Brownian motion Theorem \autoref{thm:SubordinateBrownMot}}\label{apx:ProofSubordinateBrowninanMotion}

\begin{proof}
Define
\begin{align}
   S_N &= \Sigma^N_{j=1}\epsilon_{\diadt{n}{j}} \\
   \intertext{so that}
   E[S^2_N] &=\Sigma^N_{j=1} E[\epsilon_{\diadt{n}{j}}^2] = N\sigma^2 \\
   \intertext{Without loss of generality, normalize $\epsilon$ with $\frac{\epsilon}{\sigma}$ so that we have $E[\epsilon^2] = 1$ and}
   E[\bigl(\frac{S_N}{\sqrt{N}}\bigr)^2] &= 1
\end{align}
For $\diadt{n}{j} \leq t < \diadt{n}{{j+1}}$ let
\begin{align}
  &\epsilon^{(n)}_t = \frac{1}{\sqrt{n}}[S_{[nt]} + (nt-[nt])\epsilon_{[nt]+1} \\
  \intertext{where $[nt]$ is the integer part of $nt$. So that}
  &\epsilon^{(n)}_{t+2^{-n}}-\epsilon^{(n)}_t = \frac{1}{\sqrt{n}}\bigl[ S_{[nt + n.2^{-n}]} + \tag*{} \\
  &(nt + n.2^{-n}-[nt + n.2^{-n}])\epsilon_{[nt + n.2^{-n}]+1}\bigr] - [S_{[nt]} + (nt-[nt])\epsilon_{[nt]+1}\\
  &= \Sigma^{[nt+n.2^{-n}]}_{j=[nt]+1}\epsilon_{\diadt{n}{j}} + \tag*{} \\
  &(nt + n.2^{-n}-[nt + n.2^{-n}])\epsilon_{[nt + n.2^{-n}]+1}-(nt-[nt])\epsilon_{[nt]+1} \\
  \intertext{Which implies}
  &E\bigl[(\epsilon^{(n)}_{t+2^{-n}}-\epsilon^{(n)}_t)^2 \lvert \mathcal{F}_{\diadt{n}{j}} \bigr] = [nt +n.2^{-n}] - [nt]-1 + \tag*{} \\ 
  &(nt + n.2^{-n}-[nt + n.2^{-n}])^2 + (nt-[nt])^2 \\
  &= n.2^{-n} + o(1+n^{-1}) - 1 \label{eq:BrownianFactor}\\
  \intertext{This implies that}
  &E\bigl [\{\frac{1}{\sqrt{n}}(\epsilon^{(n)}_{t+2^{-n}}-\epsilon^{(n)}_t)\}^2 \lvert \mathcal{F}_{\diadt{n}{j}} \bigr] \tag*{}\\
  &= 2^{-n}+o(n^{-1}+n^{-2})-\frac{1}{n} = c(n).2^{-n} \label{eq:BrownianScaleFactor}
\end{align}
for some monotone increasing function $c(\cdot)$. See e.g., \cite{Cadogan2011}.

To complete the proof of Theorem \autoref{thm:SubordinateBrownMot}, we note that according to precepts of construction of Brownian motion, Brownian scaling, see e.g., \cite[Thm.~4.17,~pg.~67;~and~Lemma~9.4,~pg.~104]{KaratzasShreve1991}, and Lemma B.1 in \cite{Cadogan2011d}, the quantity $\epsilon^{(n)}_{t+2^{-n}}-\epsilon^{(n)}_t$ is a \emph{scaled} Brownian motion $W(c(n)2^{-n})$for some \emph{monotone increasing} pre-subordinator `function $0\leq c(\cdot) \leq 1$.
\end{proof}

\bibliographystyle{chicago}

\bibliography{JSM2011_Proc_PaperOnAlpha}         

\begin{thebibliography}{}

\bibitem[\protect\citeauthoryear{Admati and Pfleiderer}{Admati and
  Pfleiderer}{1988}]{AdmatiPfleiderer1988}
Admati, A. and P.~Pfleiderer (1988).
\newblock {A Theory of Intraday Patterns: Volume and Price Variability}.
\newblock {\em Review of Financial Studies\/}~{\em 1\/}(1), 3--40.

\bibitem[\protect\citeauthoryear{Agarwal and Naik}{Agarwal and
  Naik}{2004}]{AgrawalNaik2004}
Agarwal, V. and N.~Naik (2004).
\newblock {Risks and Portfolio Decisions Involving Hedge Funds}.
\newblock {\em Review of Financial Studies\/}~{\em 17\/}(1), 63--98.

\bibitem[\protect\citeauthoryear{Avery, Chevallier, and Zeckhauser}{Avery
  et~al.}{2011}]{AveryChevalierZeckhauser2011}
Avery, C., J.~A. Chevallier, and R.~A. Zeckhauser (2011, Aug).
\newblock {The ''CAPS" Prediction System and Stock market Returns}.
\newblock NBER Working Paper No. 17298. Available at \emph{SSRN eLibrary}:
  \href{http://ssrn.com/abstract=1918237}{http://ssrn.com/abstract=1918237}.

\bibitem[\protect\citeauthoryear{Bakshi, Panayotov, and Skoulakis}{Bakshi
  et~al.}{2010}]{BakshiPanayotovSkoulakis2010}
Bakshi, G., G.~Panayotov, and G.~Skoulakis (2010, March).
\newblock {Improving the Predictability of Real Economic Activity and Asset
  Returns with Forward Variances Inferred from Options Portfolios}.
\newblock \emph{Working Paper No. RHS-06-136}.
\newblock Department of Finance, Robert H. Smith School of Business, Univ.
  Maryland. Available at
  \href{http://ssrn.com/abstract=1622088}{http://ssrn.com/abstract=1622088}.

\bibitem[\protect\citeauthoryear{Benes, Shepp, and Witsenhausen}{Benes
  et~al.}{1980}]{BenesSheppWitsen1980}
Benes, V.~E., L.~A. Shepp, and H.~Witsenhausen (1980).
\newblock {Some solvable stochastic control problems}.
\newblock {\em Stochastics\/}~{\em 4}, 39--83.

\bibitem[\protect\citeauthoryear{Black}{Black}{1976}]{Black1976}
Black, F. (1976).
\newblock {Studies of stock market volatility changes}.
\newblock {\em Proceedings of the American Statistical Association\/},
  177--181.
\newblock Business and Ecoonomics Section.

\bibitem[\protect\citeauthoryear{Bollen and Busse}{Bollen and
  Busse}{2001}]{BollenBusse2001}
Bollen, P.~B. and J.~A. Busse (2001).
\newblock {On The Timing Ability of Mutual Fund Managers}.
\newblock {\em Journal of Finance\/}~{\em 56}, 1075--1094.

\bibitem[\protect\citeauthoryear{Braun, Nelson, and Sunier}{Braun
  et~al.}{1995}]{BraunNelsonSunier1995}
Braun, P.~A., D.~B. Nelson, and A.~Sunier (1995, Dec.).
\newblock {Good News, Bad News, Volatility, and Betas}.
\newblock {\em Journal of Finance\/}~{\em 50\/}(5), 1575--1603.

\bibitem[\protect\citeauthoryear{Breiman}{Breiman}{1968}]{Breiman1968}
Breiman, L. (1968).
\newblock {\em Probability\/} (Unabridged reprint ed.).
\newblock SIAM Classics in Applied Mathematics. Philadelphia, PA:
  Addison-Weseley Publishing Co., Inc.; Reading, MA.

\bibitem[\protect\citeauthoryear{Brogaard}{Brogaard}{2010}]{Brogaard2010}
Brogaard, J.~A. (2010, Nov.).
\newblock {High Frequency Trading and Its Impact On Market Quality}.
\newblock Working Paper, Kellog School of Management, Department of Finance,
  Northwestern University. Available at SSRN eLibrary
  \href{http://ssrn.com/abstract=1641387}{http://ssrn.com/abstract=1641387}.

\bibitem[\protect\citeauthoryear{Cadogan}{Cadogan}{2011a}]{Cadogan2011d}
Cadogan, G. (2011a, May).
\newblock {Alpha Representation For Active Portfolio Management and High
  Frequency Trading In Seemingly Efficient Markets}.
\newblock Working Paper, \emph{unpublished}.

\bibitem[\protect\citeauthoryear{Cadogan}{Cadogan}{2011b}]{Cadogan2011}
Cadogan, G. (2011b, January).
\newblock {Canonical Representation Of Option Prices and Greeks with
  Implications For Market Timing}.
\newblock Working Paper. Available at \emph{SSRN eLibrary}
  \href{http://ssrn.com/paper=1625835}{http://ssrn.com/paper=1625835}.

\bibitem[\protect\citeauthoryear{Chance and Hemler}{Chance and
  Hemler}{2001}]{ChanceHemler2001}
Chance, D.~M. and M.~L. Hemler (2001).
\newblock {The Performance of Professional Market Timers: Daily Evidence From
  Executed Strategies}.
\newblock {\em Journal of Financial Economics\/}~{\em 62\/}(2), 377--411.

\bibitem[\protect\citeauthoryear{Christopherson, Ferson, and
  Glassman}{Christopherson et~al.}{1998}]{ChristophFersonGlassman1998}
Christopherson, J.~A., W.~A. Ferson, and D.~A. Glassman (1998, Spring).
\newblock {Conditional Manager Alphas On Economic Information: Another Look At
  The Persistence Of Performance}.
\newblock {\em Review of Financial Studies\/}~{\em 11\/}(1), 111--142.

\bibitem[\protect\citeauthoryear{Clark}{Clark}{1971}]{Clark1971}
Clark, A. (1971).
\newblock {\em Elements of Abstract Algebra\/} (Dover reprint 1984 ed.).
\newblock Belmont CA: Wadsworth Publishing Co.

\bibitem[\protect\citeauthoryear{Dellacherie and Meyer}{Dellacherie and
  Meyer}{1982}]{DellacherieMeyer1982}
Dellacherie, C. and P.~Meyer (1982).
\newblock {\em Probabilities and Potential B: Theory of Martingales}, Volume~72
  of {\em North-Holland Mathematics Studies}.
\newblock New York: North-Holland Publishing Co.

\bibitem[\protect\citeauthoryear{Doob}{Doob}{1953}]{Doob1953}
Doob, J. (1953).
\newblock {\em Stochastic Processes}.
\newblock New York, N.Y.: John Wiley \& Sons, Inc.

\bibitem[\protect\citeauthoryear{Dudley}{Dudley}{2004}]{Dudley2004}
Dudley, R.~M. (2004).
\newblock {\em Real Analysis and Probability\/} (Paperback ed.).
\newblock Number~74 in Cambridge Studies in Advanced Mathematics. New York,
  N.Y.: Cambridge University Press.

\bibitem[\protect\citeauthoryear{Elton, Gruber, Brown, and Goetzmann}{Elton
  et~al.}{2003}]{EltonGruberBrownGoetz2003}
Elton, E. .~J., M.~J. Gruber, S.~J. Brown, and W.~N. Goetzmann (2003).
\newblock {\em Modern Portfolio Theory and Investment Analysis\/} (6th ed.).
\newblock New York, N.Y.: John Wiley \& Sons, Inc.

\bibitem[\protect\citeauthoryear{Erten, Zandon\'{a}-Schneider, and
  Reigber}{Erten et~al.}{2009}]{ErtenZandonaReigber2009}
Erten, E., R.~Zandon\'{a}-Schneider, and A.~Reigber (2009).
\newblock {Statistical Characterization of the maximal Eigenvalue of A Wishart
  Distribution with Application to Multichannel SAR Systems}.
\newblock unpublished, available at
  \href{http://elib.dlr.de/59127/1/maxeg_polinsar09.pdf}{http://elib.dlr.de/59%
127/1/maxeg_polinsar09.pdf}.

\bibitem[\protect\citeauthoryear{Fama}{Fama}{1972}]{Fama1972}
Fama, E. (1972, June).
\newblock {Components of Investment Performance}.
\newblock {\em Journal of Finance\/}~{\em 27\/}(3), 551--567.

\bibitem[\protect\citeauthoryear{Fama and French}{Fama and
  French}{1993}]{FamaFrench1993}
Fama, E. and K.~French (1993).
\newblock {Common Risk Factors in the Return on Bonds and Stocks}.
\newblock {\em Journal of Financial Economics\/}~{\em 33\/}(1), 3--53.

\bibitem[\protect\citeauthoryear{Fama and French}{Fama and
  French}{1996}]{FamaFrench1996}
Fama, E. and K.~French (1996).
\newblock {Multifactor Explanations of Asset Pricing Anomalies}.
\newblock {\em Journal of Finance\/}~{\em 51\/}(1), 55--84.

\bibitem[\protect\citeauthoryear{Ferson and Schadt}{Ferson and
  Schadt}{1996}]{FersonSchadt1996}
Ferson, W. and R.~W. Schadt (1996).
\newblock {Measuring Fund Strategy and Performance in Changing Economic
  Conditions}.
\newblock {\em Journal of Finance\/}~{\em 51\/}(2), 425--461.

\bibitem[\protect\citeauthoryear{Ferson, Sarkissian, and Simin}{Ferson
  et~al.}{2003}]{FersonSarkSimin2003}
Ferson, W.~E., S.~Sarkissian, and T.~T. Simin (2003).
\newblock {Spurious Regressions in Financial Economics?}
\newblock {\em Journal of Finance\/}~{\em 58}, 1393--1414.

\bibitem[\protect\citeauthoryear{Fulkerson, Jordan, and Smith}{Fulkerson
  et~al.}{2010}]{FulkJordanSmith2010}
Fulkerson, J.~A., B.~D. Jordan, and J.~M. Smith (2010, Dec.).
\newblock {Do Short Sellers Make Money}.
\newblock Working Paper, Department of Finance, Gatton College of Business and
  Economics.

\bibitem[\protect\citeauthoryear{Fung and Hsieh}{Fung and
  Hsieh}{1997}]{FungHsieh1997}
Fung, W. and D.~A. Hsieh (1997).
\newblock {Emperical Charaacteristics Of Dynamic Tradingg Strategies}.
\newblock {\em Review of Financial Studies\/}~{\em 10\/}(2), 275--302.

\bibitem[\protect\citeauthoryear{Gikhman and Skorokhod}{Gikhman and
  Skorokhod}{1969}]{GikhmanSkorokhod1969}
Gikhman, I.~I. and A.~V. Skorokhod (1969).
\newblock {\em Introduction to The Theory of Random Processes}.
\newblock Phildelphia, PA: W. B. Saunders, Co.
\newblock Dover reprint 1996.

\bibitem[\protect\citeauthoryear{Glosten and Jagannathan}{Glosten and
  Jagannathan}{1994}]{GlostenJagannathan1994}
Glosten, L.~R. and R.~Jagannathan (1994).
\newblock {A Contingent Claim Approach To Perrformance Evaluation}.
\newblock {\em Journal of Empirical Finance\/}~{\em 1}, 133--160.

\bibitem[\protect\citeauthoryear{Greene}{Greene}{2003}]{Greene2003}
Greene, W.~H. (2003).
\newblock {\em Econometric Analysis\/} (5th ed.).
\newblock Upper Saddle Rd., N. J.: Prentice-Hall, Inc.

\bibitem[\protect\citeauthoryear{Grinblatt and Titman}{Grinblatt and
  Titman}{1989}]{GrinblattTitman1989}
Grinblatt, M. and S.~Titman (1989).
\newblock {Portfolio Performance Evaluation: Old Issues and New Insights}.
\newblock {\em Refiew of Financial Studies\/}~{\em 2\/}(3), 393--421.

\bibitem[\protect\citeauthoryear{Grinold}{Grinold}{1993}]{Grinold1993}
Grinold, R.~C. (1993, July-Aug).
\newblock {Is Beta Dead Again?}
\newblock {\em Financial Analyst Journal\/}~{\em 49\/}(4), 28--34.

\bibitem[\protect\citeauthoryear{Grinold and Kahn}{Grinold and
  Kahn}{2000}]{GrinoldKahn2000}
Grinold, R.~C. and R.~N. Kahn (2000).
\newblock {\em Active Portfolio Management: A Quantitative Approach for
  Providing Superior Returns and Controlling Risk\/} (2nd ed.).
\newblock New York: McGraw-Hill, Inc.

\bibitem[\protect\citeauthoryear{Grundy and Malkiel}{Grundy and
  Malkiel}{1996}]{GrundyMalkiel1996}
Grundy, K. and B.~G. Malkiel (1996, Feb).
\newblock {Reports of Beta's Death Have Been Greatly Exaggerated}.
\newblock {\em Journal of Portfolio Management\/}~{\em 22\/}(3), 36--44.

\bibitem[\protect\citeauthoryear{Hansen and Scheinkman}{Hansen and
  Scheinkman}{2009}]{HansenScheinkman2009}
Hansen, L.~P. and J.~A. Scheinkman (2009, Jan.).
\newblock {Long Term Risk: An Operator Approach}.
\newblock {\em Econometrica\/}~{\em 77\/}(1), 177--234.

\bibitem[\protect\citeauthoryear{Henricksson}{Henricksson}{1984}]{Henricksson1%
984}
Henricksson, R.~D. (1984, Jan.).
\newblock {Market Timing and Mutual Fund Performance: An Empirical
  Investigation}.
\newblock {\em Journal of Business\/}~{\em 57\/}(1), 73--96.

\bibitem[\protect\citeauthoryear{Henriksson and Merton}{Henriksson and
  Merton}{1981}]{HenrikMerton1981}
Henriksson, R.~D. and R.~C. Merton (1981, Oct).
\newblock {On Market Timing and Investment Performance {II}: Statistical
  Procedures for Evaluating Forecasting Skills}.
\newblock {\em Journal of Business\/}~{\em 54\/}(4), 513--533.

\bibitem[\protect\citeauthoryear{James and Stein}{James and
  Stein}{1961}]{JamesStein1961}
James, W. and C.~Stein (1961).
\newblock {Estimation with Quadratic Loss}.
\newblock In {\em Proc. 4th Berkeley Symp. on Math. Stat. and Prob.}, Volume~1,
  Berkeley, CA, pp.\  361--379. Univ. California Press, available at
  \href{http://projecteuclid.org/euclid.bsmsp/1200512173}{\underline{http://pr%
ojecteuclid.org/euclid.bsmsp/1200512173}}.

\bibitem[\protect\citeauthoryear{Jarrow and Protter}{Jarrow and
  Protter}{2010}]{JarrowProtter2010}
Jarrow, R. and P.~Protter (2010, April).
\newblock {Positive Alphas, Abnormal Performance, and Illusionary Arbitrage}.
\newblock Johnson School Research Paper \#19-2010, Dept. Finance, Cornell Univ.
  Available at
  \href{http://ssrn.com/abstract=1593051}{http://ssrn.com/abstract=1593051}.
  Forthcoming, \emph{Mathematical Finance}.

\bibitem[\protect\citeauthoryear{Jarrow}{Jarrow}{2010}]{Jarrow2010}
Jarrow, R.~A. (2010, Summer).
\newblock {Active Portfolio Manageement and Positiive Alphas: Fact or fantasy?}
\newblock {\em Journal of Portfolio Management\/}~{\em 36\/}(4), 17--22.

\bibitem[\protect\citeauthoryear{Jensen}{Jensen}{1967}]{Jensen1967}
Jensen, M.~C. (1967).
\newblock {The Peformance of Mutual Funds In The Period 1945-1964}.
\newblock {\em Journal of Finance\/}~{\em 23\/}(2), 389--416.

\bibitem[\protect\citeauthoryear{Kacperczyk, Sialm, and Zhang}{Kacperczyk
  et~al.}{2008}]{KacpercSialmZhang2008}
Kacperczyk, M., C.~Sialm, and L.~Zhang (2008).
\newblock {Unobserved Actions of Hedge Funds}.
\newblock {\em Review of Financial Studies\/}~{\em 21\/}(6), 2379--2416.

\bibitem[\protect\citeauthoryear{Karatzas and Shreve}{Karatzas and
  Shreve}{1991}]{KaratzasShreve1991}
Karatzas, I. and S.~E. Shreve (1991).
\newblock {\em Brownian Motion and Stochastic Calculus\/} (2nd ed.).
\newblock Graduate Text in Mathematics. New York, N. Y.: Springer-Verlag.

\bibitem[\protect\citeauthoryear{Karlin and Taylor}{Karlin and
  Taylor}{1981}]{KarlinTaylor1981}
Karlin, S. and H.~M. Taylor (1981).
\newblock {\em {A Second Course in Stochastic Processes}}.
\newblock New York, NY: Academic Press, Inc.

\bibitem[\protect\citeauthoryear{Kassouf}{Kassouf}{1969}]{Kassouf1969}
Kassouf, S.~T. (1969, Oct.).
\newblock {An Econometric Model for Option Price with Implications for Investor
  Audacity}.
\newblock {\em Econometrica\/}~{\em 37\/}(4), 685--694.

\bibitem[\protect\citeauthoryear{Kung and Pohlman}{Kung and
  Pohlman}{2004}]{KungPohlman2004}
Kung, E. and L.~Pohlman (2004, Spring).
\newblock {Portable Alpha: Philosophy, process, and performance}.
\newblock {\em Journal of Portfolio Management\/}, 78--87.

\bibitem[\protect\citeauthoryear{MacKinlay and Pastor}{MacKinlay and
  Pastor}{1998}]{MacKinlayPastor1998}
MacKinlay, A.~C. and L.~Pastor (1998, July).
\newblock {Asset Pricing Models: Implications for Expected Returns and
  Portfolio Selection}.
\newblock Working Paper, Dep't. Finance, The Wharton School, Univ. Penn.,
  published in \emph{Review of Financial Studies}, 2000, 13(4):83-916.

\bibitem[\protect\citeauthoryear{Mamaysky, Spiegel, and Zhang}{Mamaysky
  et~al.}{2008}]{MamayskySpiegelZhang2008}
Mamaysky, H., M.~Spiegel, and L.~Zhang (2008).
\newblock {Estimating The Dynamics of Mutual Fund Alphas and Betas}.
\newblock {\em Review of Financial Studies\/}~{\em 21\/}(1), 233--264.

\bibitem[\protect\citeauthoryear{Merton}{Merton}{1981}]{Merton1981}
Merton, R. (1981, July).
\newblock {On Market Timing and Investment Performance {I}: An Equilibrium
  Theory of Value for Market Forecasts}.
\newblock {\em Journal of Business\/}~{\em 54\/}(3), 363--406.

\bibitem[\protect\citeauthoryear{Muirhead}{Muirhead}{2005}]{Muirhead2005}
Muirhead, R.~J. (2005).
\newblock {\em Aspects of Multivariate Statistical Theory\/} (2nd ed.).
\newblock Hoboken, N. J.: John Wiley \& Sons, Inc.

\bibitem[\protect\citeauthoryear{Myneni}{Myneni}{1992}]{Myneni1992}
Myneni, R. (1992).
\newblock {The Pricing of The American Option}.
\newblock {\em Annals of Applied Probability\/}~{\em 2\/}(1), 1--23.

\bibitem[\protect\citeauthoryear{Noehel, Wang, and Zheng}{Noehel
  et~al.}{2010}]{Nohel2010}
Noehel, T., Z.~J. Wang, and J.~Zheng (2010).
\newblock {Side-by-Side Management of Hedge Funds and Mutual Funds}.
\newblock {\em Review of Financial Studies\/}~{\em 23\/}(6), 2342--2373.

\bibitem[\protect\citeauthoryear{{\O}ksendal}{{\O}ksendal}{2003}]{Oksendal2003}
{\O}ksendal, B. (2003).
\newblock {\em Stochastic Differential Equations: An IntroductionWith
  Applications\/} (6th ed.).
\newblock Universitext. New York: Springer-Verlag.

\bibitem[\protect\citeauthoryear{Phillips}{Phillips}{1998}]{Phillips1998}
Phillips, P. C.~B. (1998, Nov.).
\newblock {New Tools for Understanding Spurious Regressions}.
\newblock {\em Econometrica\/}~{\em 66\/}(6), 1299--1325.

\bibitem[\protect\citeauthoryear{Sharpe}{Sharpe}{1964}]{Sharpe1964}
Sharpe, W.~F. (1964, Sept.).
\newblock {Capital Asset Prices; A Theory of Market Equilibrium Under
  Conditions of Risk}.
\newblock {\em Journal of Finance\/}~{\em 21\/}(3), 425--442.

\bibitem[\protect\citeauthoryear{Snell}{Snell}{1952}]{Snell1952}
Snell, J.~L. (1952, Sept.).
\newblock {Applications of Martingale System Theorems}.
\newblock {\em Trans. Amer. Math. Soc.\/}~{\em 73\/}(2), 293--312.

\bibitem[\protect\citeauthoryear{Treynor and Mazuy}{Treynor and
  Mazuy}{1966}]{TreynorMazuy1966}
Treynor, J. and K.~Mazuy (1966).
\newblock {Can Mutual Funds Outguess The Market?}
\newblock {\em Harvard Business Review\/}~{\em 44}, 131--136.

\bibitem[\protect\citeauthoryear{Treynor and Black}{Treynor and
  Black}{1973}]{TreynorBlack1973}
Treynor, J.~L. and F.~Black (1973, Jan.).
\newblock {How To Use Security Anallysis to Improve Portfolio Selection}.
\newblock {\em Journal of Business\/}~{\em 46\/}(1), 66--86.

\bibitem[\protect\citeauthoryear{Urstadt}{Urstadt}{2010}]{Urstadt2010}
Urstadt, B. (2010, Jan-Feb).
\newblock {Trading Shares in Milliseconds}.
\newblock {\em Technology Review\/}, 44--49.

\bibitem[\protect\citeauthoryear{Vervaat}{Vervaat}{1979}]{Vervaat1979}
Vervaat, W. (1979).
\newblock {A Relation Between Brownian Bridge and Brownian Excursion}.
\newblock {\em Annals of Probability\/}~{\em 7\/}(1), 143--149.

\end{thebibliography}

\end{document}